\documentclass[conference]{IEEEtran}

\usepackage[pdftex]{graphicx}
\usepackage{url}
\usepackage{flushend}
\usepackage{amsmath, amsfonts}
\usepackage{amsmath}
\usepackage{amsthm}
\usepackage{colortbl}
\usepackage{subfigure}
\usepackage{color}

\newtheorem{thm}{Theorem}

\newtheorem{prop}{Proposition}
\newtheorem{lem}{Lemma}

\begin{document}
%\sloppy

\title{Multi-Channel Random Access with Replications}
%replicas Tx decrease delay}

\author{\IEEEauthorblockN{Olga Galinina\IEEEauthorrefmark{1}, Andrey Turlikov\IEEEauthorrefmark{2}, Sergey Andreev\IEEEauthorrefmark{1}, and Yevgeni Koucheryavy\IEEEauthorrefmark{1}\\}
%\author{\IEEEauthorblockN{Author 1, Author 2, Author 3, Author 4, Author 5, and Author 6}}
 \IEEEauthorblockA{
   \IEEEauthorrefmark{1}Tampere University of Technology, Tampere, Finland}
 \IEEEauthorblockA{
   \IEEEauthorrefmark{2}State University of Aerospace Instrumentation, St. Petersburg, Russia}
}

\maketitle

\begin{abstract}
%...reliability in machine-type communications (MTC)

This paper\footnote{This work is supported by the projects TAKE-5, TT5G, and WiFiUS as well as by RFBR (research project No. 17-07-00142). 
The work of the first author is supported with a personal research grant by the Finnish Cultural Foundation and by a Jorma Ollila grant from Nokia Foundation.} considers a class of multi-channel random access algorithms, where contending devices may send multiple copies (replicas) of their messages to the central base station. We first develop a hypothetical algorithm that delivers \textit{a lower estimate} for the access delay performance within this class. Further, we propose a \textit{feasible} access control algorithm achieving low access delay by sending multiple message replicas, which approaches \textcolor{black}{the performance of the hypothetical algorithm}. The resulting performance is readily approximated by a simple lower bound, which is derived for a large number of channels.

%hypothetical, practical, simple lower bound

\end{abstract}

% \textcolor{red}{Generalized Encoding CRDSA: Maximizin Throughput in Enhanced Random Access Schemes for Satellite "cite{bacco2013generalized}"?}

\section{Introduction}
%one of the options for 5G....
%
 %
In wireless systems, random access (RA) algorithms are used primarily to arbitrate connection setup of dynamic device population over a shared communications medium. In particular, multi-channel RA algorithms have been successfully applied in modern wireless networks, where contenders may transmit on several orthogonal time-, frequency-, or code-based channels. Today, multiple variations of such solutions have been developed and thoroughly studied~\cite{choi2006multichannel}.

Recently, with the rapid proliferation of mission-aware industrial applications and the corresponding communications enablers, there is increased interest in the research community to revisit multi-channel RA algorithm design for improved transmission reliability. To facilitate reliable message delivery in machine-type communications (MTC), it has been proposed to exploit redundancy by sending multiple copies (replicas) of the same message~\cite{paolini2015coded}.%, which is expected to sufficiently improve the performance of multichannel contention-based wireless systems. 

Another pressing demand in mission-control industrial operation is to reduce the channel access delay for MTC devices. Although delay evaluation for certain multi-replica formulations was addressed earlier in the context of satellite systems with one~\cite{hajek1994delay} or several channels~\cite{choudhury1983diversity}, \textit{dynamic centralized access control} of such systems has not been attempted due to substantial feedback delay.

%\cite{choudhury1983diversity}  Diversity ALOHA--A random access scheme for satellite communications; no backlog, no control, no stability
%\cite{hajek1994delay} On the delay in a multiple-access system with large propagation delay: similar discussion, one channel with large feedback delay, Tx several copies; time-division, no control algorithm, no noise

%\cite{paolini2015coded} Coded random access: applying codes on graphs to design random access protocols
%\cite{taghavi2016design} On the Design of Universal Schemes for Massive Uncoordinated Multiple Access
% \cite{condoluci2016enhanced} Enhanced radio access and data transmission procedures facilitating industry-compliant machine-type communications over LTE-based 5G networks
%?
%{yilmaz2015analysis}  Analysis of ultra-reliable and low-latency 5G communication for a factory automation use case
%thomsen2013code  Code-expanded radio access protocol for machine-to-machine communications
%\cite{bacco2013generalized} Generalized Encoding CRDSA: Maximizing Throughput in Enhanced Random Access Schemes for Satellite
%zanella: work with estimation of conflict set

%\textcolor{red}{On the Design of Universal Schemes for Massive Uncoordinated Multiple Access? \cite{taghavi2016design}}

In this work, we focus on a class of \textit{multi-channel} wireless RA algorithms that exploit simultaneous transmission of multiple message replicas. Assuming no interference cancellation or other similar means on the physical layer (such as those in~\cite{taghavi2016design}, \cite{paolini2015coded}), we aim to analyze the number of backlogged (ready to transmit) MTC devices and thus the channel access delay. Considering error-prone radio channel, we characterize the best achievable system performance by outlining a hypothetical RA algorithm, which minimizes the collision probability at any message transmission opportunity. %\textcolor{red}{rewrite}
   
The rest of this text is organized as follows. The main assumptions of our system model are summarized in Section II. Section III formulates our hypothetical algorithm for estimating the channel access delay, which is minimal within the considered class of RA algorithms. Section IV describes feasible practical solutions, while Section V derives the lower bounds on the number of backlogged MTC devices by quantifying the respective performance limits. Finally, we offer some numerical results in Section VI and conclude.

\section{System model}
%In this section, we discuss the main assumptions of our multi-channel multi-replica system model.
\vspace{-3px}
\subsection{Main assumptions}
We consider a centralized radio access system, where MTC devices activate irregularly and attempt to establish a connection with their serving base station (BS). In doing so, they transmit short messages by employing a time-slotted contention-based algorithm. For the sake of analytical tractability, let new devices activate over time randomly, by following a Poisson process with the intensity of $\Lambda$ (i.e., an infinite population of MTC devices is assumed). Upon its activation, a particular device may decide to transmit during a slot $t$ with the probability $p_{t}$, in which case it selects $K_t \geq 1$ out of $M$ available channels for sending either a one-slot message (if $K_t = 1$) or several identical replicas of this message. Generally, while the notion of a ``channel'' may involve time-, frequency-, or code-division structure, to minimize the initial access delay we here imply either code- or frequency-based division, so that multiple message replicas are transmitted at the same time. 

At the end of a transmission time interval, the MTC devices receive error-free feedback from the BS and may also observe the outcomes of actions by other contenders. Within a particular slot $t$, we differentiate between the following events in the channel: (i) \textit{idle time} when no MTC device attempts to access the channel, (ii) \textit{potential successful replica delivery} if only one device attempts to transmit on the channel, and (iii) \textit{collision of replicas} if two or more devices choose to transmit on the channel. In the latter case, we assume that no message is received successfully, and thus all of the colliding contenders fail in their replica transmissions. 

Further, collisions in our model are not the only source of replica transmission failures. We also assume that the radio channel is error-prone: the replicas that have avoided a collision may still be received incorrectly by the BS with the probability of $\gamma$, and therefore be considered lost. In this work, we do not account for any interference cancellation mechanisms or other packet recovery procedures on the physical layer. Finally, if all of $K_t$ replicas sent by a certain MTC device are lost, the transmission of the corresponding message fails, whereas the device in question becomes \textit{backlogged} and attempts to retransmit in the consecutive slot $t+1$ with the probability of $p_{t+1}$ and the number of replicas $K_{t+1}$. If however at least one of the replicas is transmitted successfully, then the message is assumed to be delivered and the respective MTC device deactivates permanently.

The (re)transmission probability $p_t$ and the allowed number of replicas $K_t$ may be announced by the BS and affect the resulting channel access delay as well as the overall system performance. 
%For the sake of exposition, we assume that the initial transmission probability is the same as the retransmission probability, i.e., $q_t = p_t$.
The appropriate choice of these two system control parameters becomes the focus of our subsequent study.%: (i) the (re)transmission probability $p_t$, and (ii) the allowed number of replicas $K_t$. Both parameters are announced by the BS.

\subsection{Possible extensions to the model}
We note that in some practical systems, the feedback from the BS may be delayed and/or imperfect, which may require the consideration of a feedback interval as well as a certain feedback transmission failure probability. Both of these formulations constitute separate research problems, which could be addressed by extending our proposed approach. 

In what follows, in order to derive our practical control algorithm we assume that the BS observes and may thus estimate separately (i) the number of idle channels $i$, (ii) the number of replica transmissions that have successfully avoided collisions $s$, (iii) the number of channels that experience a collision $c$ ($i+s+c = M$), as well as (iv) the total number of successful MTC devices $N_s$. In some circumstances, the BS may not be capable of differentiating between the collisions and the transmissions without collisions that are however lost due to the error-prone channel. This may generally depend on the properties of the physical layer implementation in the communications equipment. In this situation, one may readily reformulate our proposed control scheme accordingly, but the resulting performance of such less informed system would naturally degrade.

Moreover, some practical implementations may decide to dynamically adjust the number of channels that are available for the transmissions from MTC devices in slot $t$. Hence, we do not limit our formulations by assuming that the number of channels $M$ is constant. On the contrary, our target control scheme is equally suitable to capture the time-variant channel availability, where $M_t$ channels are available in slot $t$ as announced by the BS.

%-------------------------------------------------------------------------------------------------------------------------------------------------------------------------------------------------------%

\section{Hypothetical control algorithm}

In this section, we introduce a hypothetical (optimal) control procedure that guarantees the minimal channel access delay within the class of the considered RA algorithms (i.e., those controlling the pair $<p_t, K_t>$ according to the system model assumptions outlined above). For comparison, we also refer to the optimal scheme for the conventional access systems without message transmission redundancy (i.e., those controlling only $<p_t>$). For better readability, hereinafter we adopt the following naming convention. We denote the optimal algorithms which require certain hypothetically available information on the system state as \textbf{HK} and \textbf{H1}, while their feasible practical implementations proposed in Section~IV are referred to as \textbf{AK} and \textbf{A1} for multi-replica and single-replica transmission, respectively.

 %For brevity, hereinafter we omit the index $t$.
%\subsection{Optimal algorithms}
In order to derive the optimal algorithms that deliver the lower delay estimates, let us assume that for every slot $t$ the number of MTC devices $N_t$, which have activated and decided to transmit, is known to the BS. For the single-replica systems, the optimal control scheme (referred to as the \textit{hypothetical algorithm H1} in what follows) may be formulated as Proposition~1 adapted from~\cite{galinina2013stabilizing}:

\begin{prop} \textbf{Algorithm H1.}
For the system with no redundancy ($K_t = 1$), the optimal algorithm H1 that minimizes the access delay corresponds to the transmission probability $p_t = \min \left(1,{M}/{N_t}\right)$, where $N_t$ is the number of currently active MTC devices and $M$ is the number of channels.
\end{prop}

Importantly, the above control procedure delivers the \textit{lower estimate} for the channel access delay within the class of control algorithms with no redundancy. We also emphasize that the conventional algorithm H1 is optimal for \textit{both} the error-free channel case with $\gamma = 0$ and the error-prone channel case with $\gamma>0$. 

Further, if $N_t<M$, then sending multiple replicas may increase the probability of successful replica reception at every transmission opportunity. Given the number of transmitting MTC devices $N_t$, the number of channels $M$, and by estimating the replica corruption probability $\gamma$, we may derive the optimal number of replicas that maximizes the probability of successful message delivery over $K$ for a given $N_t,M$ as:
\begin{equation}
\vspace{-5px}
\begin{array}{c}
K^{*} = \arg \max_{K} p_s(N_t,M,\gamma, K), 
\end{array} \label{eqn:opt}
\end{equation}
where $p_s(N_t,M,\gamma, K) = \Pr\{\text{successful message delivery}\}$.

Hence, for \textit{multi-replica} systems where $K_t \geq 1$, we employ a similar technique as has been used for the conventional systems with no redundancy. Therefore, we establish the optimal control parameters in question as (i) $K_t=1$, $p_t = {M}/{N_t}$ if $N_t>M$, or, otherwise, (ii) $p_t=1$, $K_t = K^* = f(N_t,M,\gamma)$, where \textcolor{black}{the solution to the optimization problem (\ref{eqn:opt}) is the value taken from the preset table $f(N_t,M,\gamma)$, which can be provisionally calculated and stored in the BS. The derivations necessary for solving (\ref{eqn:opt}) (i.e., optimizing the successful delivery process) 
are summarized in Appendix}. Given all of the above, we may formulate our Proposition~2.\vspace{-5px}
\begin{prop} \textbf{Algorithm HK.}
For the multi-replica system with redundancy ($K_t \geq 1$), the optimal algorithm HK that minimizes the channel access delay returns the following set of control parameters: $p_t = \min\left(1,{M}/{N_t}\right)$, $K_t = f(N_t,M,\gamma)$.
\end{prop}
\begin{proof}
The fact that the considered multi-replica algorithm HK outperforms all of the other algorithms, which control the system through the pair $<p_t,K_t>$, follows from the choice of $K^*$ that is based on the highest message delivery probability.
\end{proof}
%FIXME (DONT DELETE) \textcolor{red}{do we need to prove that for $K>1$ $p=1$?}

The following Theorem~1 compares any single-replica channel access scheme and the algorithm HK (further illustrative examples for Theorem~1 are provided in Section~VI). 
\begin{thm}
The optimal multi-replica algorithm HK outperforms the algorithm H1 and thus any other single-replica algorithm in terms of the channel access delay.% \textcolor{blue}{within the region $\lambda < e^{-1}$}.
\end{thm}
\begin{proof}
Immediately follows from Proposition~1 and Proposition~2.
\end{proof}

We note however that in practice the exact number of active MTC devices is unknown, and thus all the feasible channel access algorithms will necessarily result in higher delay than the one that the optimal algorithm can achieve. Our further goal is to propose such a practical scheme, which would be comparable to the lower delay estimate and yet technically feasible to implement in the real wireless systems.

%-------------------------------------------------------------------------------------------------------------------------------------------------------------------------------------------------------%

%-------------------------------------------------------------------------------------------------------------------------------------------------------------------------------------------------------%

\section{Practical control algorithm}

In this section, we detail the practical control procedures for each of the two considered classes. We remind that for estimating the unknown number $N_t$, the BS may observe and exploit the following four parameters: (i) the number of idle channels $i_t$ that no MTC device selects, (ii) the number of busy channels $s_t$ that experience one transmission, (iii) the number of busy channels $c_t$ that experience a collision, and (iv) the number of successfully delivered messages $N_s$. 
%Here, $I\{N_{t,j}\} $ is the indicator function of channel events.  = \sum^{M}_{j=1} I\{N_{t,j}=1\}
Further, we assume that the system load per channel $\lambda = {\Lambda}/{M}$ may also be estimated at the BS side. 

First, let us describe a feasible practical implementation for the class of algorithms with single-replica transmission~\cite{galinina2013stabilizing} (i.e., Algorithm A1 in our notation).
\begin{prop} \textbf{Algorithm A1.}
For the system with no redundancy ($K_t = 1$), one of the feasible implementations may be derived through an auxiliary random process $Z_t$ and the corresponding probability $p_t = \min\left(1,{M}/{Z_t}\right)$, where $M$ is the number of channels and $Z_t$ adheres to the following evolution:
\begin{equation}
\begin{array}{c}
Z_{t+1}=\max\{1,Z_{t}+\Delta Z_{t}\},\\
\Delta Z_t = a \cdot i_t + b\cdot s_t + c\cdot c_t,
\end{array}
\end{equation}
while the constants $a<0,b,c>0$ are connected via the expression $c(e-2)+a+b = 0$~\cite{galinina2013stabilizing} and guarantee stability of the algorithm A1 for all $\lambda<e^{-1}$.
\end{prop}

We note that while yielding the lower delay estimate within its class, the algorithm H1 outperforms A1 or any other alternative procedure with no redundancy, and so does the algorithm HK. We further aim at developing a practical algorithm for the formulation with redundancy that would result in a lower delay than what A1 delivers. In order to do so, we exploit the situations when adding redundancy improves the probability of successful message transmission, i.e., $p_t = 1$, $K_t = K^*$. In contrast to the hypothetical optimal scheme, information on $N_t$ is unavailable here and needs to be adequately estimated as $\tilde N_t$, so that the predefined table $f(\tilde N_t,M,\gamma)$ would yield $K^*$, which is close to the optimal value.

For this purpose, let us apply the method of maximum likelihood, which allows us to select the most probable event $N$ based on the observed outcome $\mathbf{v} = <i,s,c>$ (for clarity, hereinafter we omit the indexes $t$). An appropriate estimation is based on the probability distribution $\Pr\{\mathbf{v}| N\}$ and appears to be rather cumbersome for the required `on-the-fly' optimization process. Instead, following the approach in~\cite{zanella2012estimating}, we employ a simpler approximation:
\begin{lem}
The point $N$ that delivers the maximum of likelihood may be tightly approximated by the point of maximum of the function:
\begin{equation}
\begin{array}{c}
 g(\mathbf{v},N) = \mu(N)^s e^{-\mu(N) M} \left( e^{\mu(N)}-1-\mu(N)\right)^c,
\end{array}
\end{equation}
where $\mu(N) = {NK}/{M}$ while $i$, $s$, and $c$ are the numbers of idle channels, successful replicas, and collisions, respectively. 
\end{lem}
\begin{proof}
Assuming that the number of channels is large enough, we may approximate the number of replicas that ``choose'' the channels by independent random variables distributed according to a Poisson process with the intensity of ${N K}/{M}$. Using this Poisson approximation, we may estimate the probabilities to observe $i$ idle channels, $s$ successful replicas, and $c$ collisions as $e^{-\mu  \cdot i}$, $\left(\mu e^{-\mu} \right)^s$, and $(1-e^{-\mu}-\mu e^{-\mu})^c$, correspondingly. Bringing all the three together, we may write:
\begin{equation}
\begin{array}{c}
\Pr\{\mathbf{v}| N\}  \approx C\mu^s e^{-\mu (i+s+c)} \left( e^{\mu}-1-\mu\right)^c, \label{eqn:probability}
\end{array}
\end{equation}
where $C$ is the constant number of combinations for a given set $<i,s,c>$. The maximum of (\ref{eqn:probability}) is located at the same point as the maximum of $g(\mathbf{v},N)$. 
\end{proof}

We emphasize that $g(\mathbf{v},N)$ resembles the shape of $\Pr\{\mathbf{v}| N\}$ (which is sufficient for the purposes of our optimization), but does not have the meaning of probability. The optimum for $c>0$ may be obtained similarly to~\cite{zanella2012estimating} as a numerical solution to the following equation (follows from the zero gradient $g'_{\mu}(\mathbf{v},N)=0$):
\begin{equation}
\begin{array}{c}
 c \mu \left( e^{\mu}-1\right)-\left( \mu M-s\right)\left( e^{\mu}-1-\mu\right) = 0.  
 \end{array}\label{eqn:zerogradient}
\end{equation}

Hence, given the point $\mu^*$ obtained from the expression (\ref{eqn:zerogradient}) and the number of successful replicas $s$, the number of currently active MTC devices may be estimated as follows:
\begin{equation}
 \tilde N_{t} = \left \{
 \begin{array}{c}
\left [ \frac{1}{p} \mu^* \frac{M}{K} + \lambda M \right] - N_s, \quad \text{for $c>0$}\\
\left[ \frac{1}{p} s \frac{1}{K}+ \lambda M\right] - N_s, \quad \text{for $c=0$},
 \end{array} 
 \right. \label{eqn:algHK}
\end{equation}
where $p = p_{t-1}$ is the transmission probability and $K = K_{t-1}$ is the number of replicas in the previous slot. Based on the obtained estimator, we may finally formulate the sought practical algorithm that exploits message transmission redundancy.

\begin{prop}\textbf{Algorithm AK.}
For the system with the use of redundancy ($K_t \geq 1$), one of the practical algorithms may be constructed on the set of control parameters: $p_{t} = \min \left(1,{M}/{\tilde N_{t}} \right)$, $K_{t} = K^* = f(\tilde N_{t},M,\gamma)$, where $\tilde N_{t}$ is delivered by the expression (\ref{eqn:algHK}).
\end{prop}

We have established that contrary to the provably stable algorithm A1, \textcolor{black}{the algorithm AK may demonstrate instability due to the usage of the estimate $\tilde N$ (rigorous proof of such instability is a stand-alone problem out of scope of this paper), and hence may result in unpredictably high delay values for the range of loads that are close to the multi-channel system capacity} $Ke^{-1}$~\cite{galinina2013stabilizing}. To avoid that, we propose taking advantage of the algorithm A1 by applying $p_t = \min \left(1,{M}/{Z_t}\right)$ whenever the estimated number of transmitting MTC devices increases.% to $\tilde N_t > M $. 

%\begin{prop} \textcolor{red}{check the name of this?}{}
\textcolor{black}{
\textbf{Proposition 4a.} \textbf{Modified algorithm AK.} 
%\textit{The resulting \emph{combined algorithm}:
\textit{Algorithm \vspace{-5px}
\begin{equation}
\vspace{-5px}
\left\{
\begin{array}{l}
p_{t} = 1,\quad K_{t} = f(\tilde N_{t},M,\gamma),\quad\text{ if $\tilde N_{t}<M$},\\
p_{t} = \min \left(1,{M}/{Z_t} \right), \quad K_{t} = 1, \quad\text{ if $\tilde N_{t} \geq M$},
\end{array} \right.
\end{equation}
delivers lower channel access delay due to the use of redundancy, but remains \textit{stable} at higher loads. Stability follows from the properties of the algorithm in~\cite{galinina2013stabilizing}, particularly, from the ergodicity of two-dimensional process $(Z_t,N_t)$. }}
%\end{prop}

Interestingly, our proposed estimator becomes tighter with the increasing number of channels, since the Poisson approximation (as demonstrated in Section~VI) becomes more exact when $M \rightarrow \infty$. The corresponding expressions for the performance limits are derived in the following section.

%-------------------------------------------------------------------------------------------------------------------------------------------------------------------------------------------------------%

\section{Limiting expressions}

In this section, we study the performance limits under $M \rightarrow \infty$, which approximate the numerical results for a higher number of channels, while constituting the \textit{ultimate lower bound} for both classes of algorithms.

We begin with considering the algorithm H1 and formulate the following Theorem for the limit \textcolor{black}{on the average number} of backlogged MTC devices in the system.% (i.e., those devices that retransmit after a collision).
\begin{thm}\textbf{Algorithm H1}
For the system with no redundancy ($K_t = 1$), given the channel corruption probability $\gamma$ and the load per channel $\lambda < e^{-1}$, the limit on the average number of backlogged MTC devices per channel for $M \rightarrow \infty$ may be obtained as follows:\vspace{-10px}
\begin{equation}
\vspace{-5px}
\begin{array}{c}
\eta^* = -W\left(-\frac{\lambda}{1-\gamma}\right)-\lambda,\label{eqn:alg_H1}
\end{array}
\end{equation}
where $W(x)$ is the Lambert function, which is a solution to the transcendent equation $\eta e^{\eta} = x$. 
\end{thm}
\begin{proof}
As in the previous section, we employ the Poisson approximation by assuming that for the large values of $M$ in the stationary mode the flows of newly activated and backlogged MTC devices form a Poisson process with some intensity $\eta$. The probability that one transmitted replica avoids a collision in the selected channel is therefore defined by $e^{-\eta}$. In the stationary mode, the output flow equals the system load and $\eta e^{-\eta} (1-\gamma)= \lambda$. The sought solution may thus be found as $\eta^* = -W\left(-\frac{\lambda}{1-\gamma}\right)$, which however contains the share $\lambda$ of newly arriving MTC devices. By subtracting the latter from $\eta^*$, we establish the target expression (\ref{eqn:alg_H1}).
\end{proof}
Following the same logic for the limit on the backlogged MTC devices imposed by the algorithm HK, we arrive at the following Theorem.\vspace{-5px}
\begin{thm}\textbf{Algorithm HK}
For the system with the use of redundancy ($K_t \geq 1$), given the channel corruption probability $\gamma$ and the load per channel $\lambda < e^{-1}$, the average target number of replicas $K^*$ and the \textcolor{black}{limit on the average number} of the backlogged MTC devices for $M \rightarrow \infty$ are given as:\vspace{-5px}
\begin{equation}\vspace{-5px}
\begin{array}{c}
K^* = \arg \max_{K} h(K), \quad \eta^* = h(K^*)-\lambda,
%h(K) = -W\left(-\frac{\lambda}{1-\gamma}\right)-\lambda,
\end{array} \vspace{-3px}
\end{equation}
where $\eta = h(K)$ is solution to the equation:\vspace{-5px}
\begin{equation}\vspace{-5px}
\begin{array}{c}
\lambda = \eta \left[1-\left(1-(1-\gamma)  e^{-K\eta}   \right)^K \right].
\end{array}
\end{equation}
 %\label{eqn:alg_HK}
\end{thm}
\begin{proof}
Similarly, the flows of newly activated and backlogged MTC devices may be considered equal to $\eta$, such that the corresponding intensity of the flow of replicas equals $K\eta$. The probability that one transmitted replica is delivered successfully can be established as $(1-\gamma) \left(  e^{-K\eta}  \right)$, and hence the message success probability equals $1-\left(1-(1-\gamma)  e^{-K\eta}  \right)^K$. From the equality of the output and input flows, the following equation follows:\vspace{-5px}
\begin{equation}
\vspace{-8px}
\begin{array}{l}
\lambda = \eta \left[1-\left(1-(1-\gamma)  e^{-K\eta}  \right)^K \right]
\end{array}. \label{eqn:temp}
\end{equation}
If $\eta^*$ is the solution to the equation (\ref{eqn:temp}), then the strategy of the algorithm HK corresponds to the maximum of $\eta^*$ over all possible $K$, which results in selecting $K^* = \arg \max_{K} h(K)$, while the backlog of $\eta^* = h(K^*)-\lambda$.
\end{proof}

 %FIXME DONT DELETE \textcolor{red}{add (9) similar to (1.6) \cite{hajek1994delay}?}

%-------------------------------------------------------------------------------------------------------------------------------------------------------------------------------------------------------%

\vspace{-10px}
\section{Numerical results}
\vspace{-5px}
In this section, we compare the proposed practical control \textit{algorithm AK} against the hypothetical \textit{algorithm HK} as well as observe performance of the hypothetical \textit{algorithm H1} and the practical control \textit{algorithm A1} that do not add redundancy, to provide more insight. As a metric of interest equivalent to the channel access delay, we study the number of backlogged MTC devices per channel (i.e., the devices activated in the previous slots), which immediately translates into the waiting time through the Little's law. 

In Fig.~\ref{fig:figure1}, we illustrate the evolution of the backlog size for all four considered procedures (for $M=10$), as well as highlight the lower bounds offered in Theorems 2 and 3 (i.e., the highlighted areas). Importantly, 
%as it could be seen in Figure \ref{fig:figure1}.a 
the use of multi-replica transmission in the case of error-free channel ($\gamma = 0$) reveals rather marginal performance gains (see the difference between AK and H1); however, we may still exploit the practical algorithm A1, as it demonstrates performance that is close to that of the hypothetical single-replica scheme H1. Conversely, as $\gamma$ increases, the benefits of utilizing multiple replicas become more substantial, especially in the region of low channel loads (improvement of up to 10 times). For any $\gamma$, the difference between all four algorithms diminishes when $\lambda$ grows within its stability region $\lambda<(1-\gamma)e^{-1}$, as redundancy is no longer helpful then.

%In this section, we 

\begin{figure} [!ht]% \vspace{-5px}
  \begin{center}
  \vspace{-15px}
    \includegraphics[width=0.7\columnwidth]{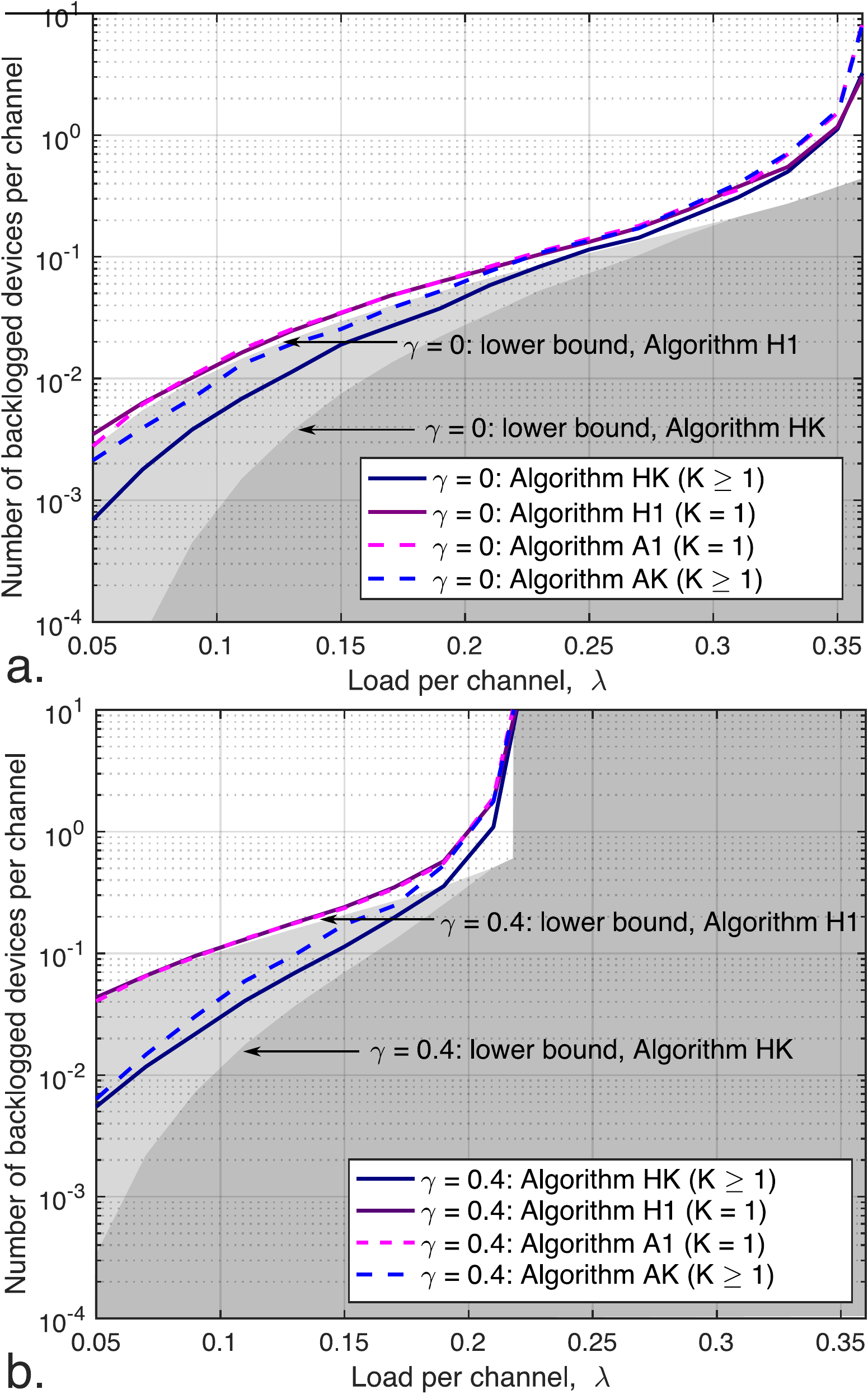}%{./Pictures/mainscreen1.png}
    \vspace{-10px}
  \end{center}
  \caption{Evolution of the number of backlogged MTC devices per channel vs. system load per channel: (a) error-free channel and (b) error-prone channel.}
  \vspace{-20px}
  \label{fig:figure1}
\end{figure} 

In Fig.~\ref{fig:figure3}, we further demonstrate the behavior of the proposed ultimate lower bounds provided by Theorems 2 and 3 in terms of the channel backlog size, as well as emphasize the average optimal numbers of replicas. Intuitively, with improved channel operation (i.e., $\gamma \rightarrow 0$), the number of backlogged MTC devices decreases down to a certain threshold. %From the Figure \ref{fig:figure3}.b we learn that higher probabilities of channel error require higher number of transmitting replicas and faster drop to $1$ by the end of region of stability $\lambda < (1-\gamma)e^{-1}$.
We note that our proposed solution HK always resides in the region between the dashed and the solid lines for any $\gamma$ and $\lambda$, while approaching the latter as $M$ increases. 

%or the increasing number of channels $M>20$, the resulting system behavior tends to our derived performance bounds, which may be easily verified for any channel load $\lambda<(1-\gamma)e^{-1}$.

\vspace{-4px}
\section{Conclusions}
\vspace{-4px}

We have shown that the channel access delay could be decreased without the use of physical layer features and is most beneficial for an error-prone channel. The proposed practical algorithm (given by Proposition 4a) makes it possible to dynamically control the system and achieve performance close to what the hypothetical algorithm may offer. At higher loads, our solution by construction becomes a provably stable algorithm for any $M$. Finally, we provide a simple lower bound for the resulting performance.

\begin{figure} [!ht]
\vspace{-10px}
  \begin{center}
    \includegraphics[width=0.8\columnwidth]{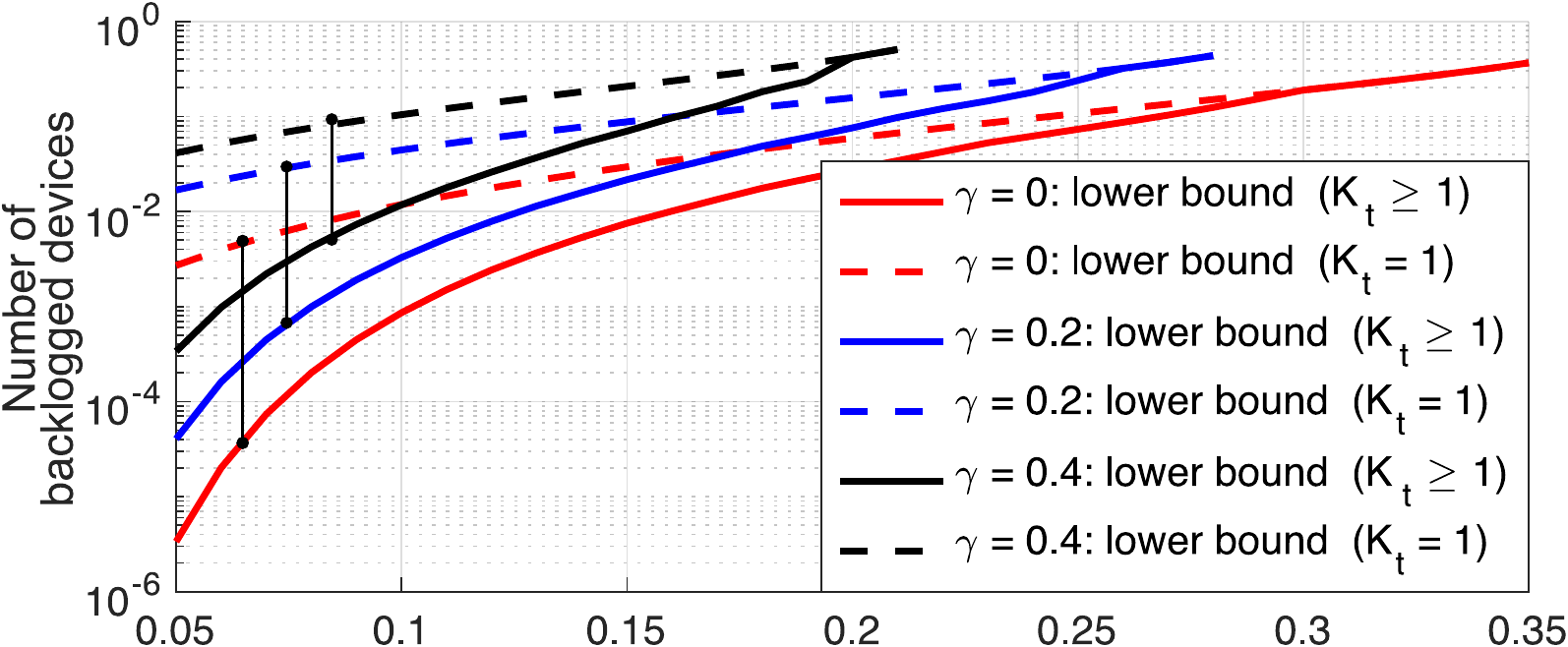}%{./Pictures/mainscreen1.png}
    \vspace{-10px}
  \end{center}
  \caption{Asymptotic performance ($M \rightarrow \infty$) of the hypothetical algorithms H1 and HK, $\gamma = 0, 0.2, 0.4$: average number of backlogged MTC devices vs. load per channel.}
  \label{fig:figure3}
  \vspace{-10px}
\end{figure} 

\vspace{-4px}
\section*{Appendix: Probability of success in one slot}
%Here, we obtain the probability of successful message delivery within one particular slot based on the known parameters $N,M,\gamma$.

%\subsection{Probability of success within one slot}

Let us consider a particular contention slot at the moment $t$ by assuming that $N$ MTC devices are currently active and may potentially transmit their replicas in this time interval. Analyzing the tagged device, which has made a decision to transmit, we may reformulate our sub-problem in terms of the classical occupation problem for the rest of $N-1$ participants. Following the standard techniques~\cite{feller1968introduction}, we may calculate 
%The number of arrangements for $N\!-\!1\!>\!0$ devices and the corresponding selected subframes may be characterized via the number of K-permutations of $L$, i.e., as the product $\left[L(L-1)...(L-K+1) \right]^{L-1} = \left[K!{{L \choose K}}\right]^{L-1} $. Let us further introduce the event $A_i$, when the subframe $i$ is \textit{not selected} by any of $N-1$ devices. Therefore, the remaining $L-1$ subframes may be arranged in $\left[(L-1)...(L-K) \right]^{(N-1)} = \left[K!{{L-1 \choose K}}\right]^{N-1} $ ways. Then, the probability of $A_i$ as well as the probability of $A_{i,j}$ (i.e., for a pair of ``empty'' subframes $i,j$) may be derived as:
%\begin{equation}
%\begin{array}{c}
%\Pr\{A_i\} =  \left[\frac{{L-1 \choose K}} { {L \choose K}}\right] ^{N-1},\quad \Pr\{A_{ij}\} =  \left[\frac{{L-2 \choose K}} { {L \choose K}}\right] ^{N-1}.
%\end{array}
%\end{equation}
%
%Continuing along these lines and taking into the account the number of such combinations, we may establish the following expressions for the probability $\Pr\left \{S_{v} \right \}$ to observe a given number $v$ of not selected subframes:
%\begin{equation}
%\begin{array}{c}
%\Pr\left \{S_{v} \right \} = {L \choose v} \left[\frac{{L-v \choose K}} { {L \choose K}}\right] ^{N-1},0\leq  v\leq L-K.
%\end{array}
%\end{equation}
%Using the , we proceed with 
the probability that exactly $M-n$ channels are occupied if there are $M-n$ channels in total (i.e., following from the standard expression for a union of the events):
\vspace{-5px}
\begin{equation}
\begin{array}{c}
\!\!\!\!\!\!p_0(M\!-\!n,N\!-\!1)\! =\!\!\!\! \sum \limits_{v=0}^{M-n-K}(-1)^v \! { {M-n} \choose v}  \left[\frac{{ {M-n-v}\choose K }}{{ {M-n} \choose K}}\right]^{N-1} \!\!\!, 
\end{array} \nonumber
\vspace{-5px}
\end{equation}
for $K\leq M-n \leq K(N-1)$ (since $K$ channels are always occupied and no more than $K(N-1)$ can be selected by $N-1$ devices), or otherwise $p_0(M-n,N\!-\!1) = 0$. 

The probability that \textit{exactly} $n$ out of $M$ channels are available (or ``empty'', i.e., do not contain any replica by other devices) or, equivalently, \textit{exactly} $M\!-\!n$ are occupied by $N\!-\!1$ devices, is given as follows:
\vspace{-5px}
\begin{equation}\vspace{-5px}
\begin{array}{c}
p_n(N\!-\!1)\! =\! 
{M \choose n} \left[ \frac{{M-n \choose K}}{{M \choose K }}\right]^{N-1} p_0(M-n,N\!-\!1)\! = 
%\\
%{M \choose n} \left[ \frac{{M-n \choose K}}{{M \choose K }}\right]^{N-1}\sum \limits_{v=0}^{M-n-K}(-1)^v { {M-n} \choose v}  \left[\frac{{ {M-n-v}\choose K }}{{ {M-n} \choose K}}\right]^{N-1}=
\\
{M \choose n} \sum \limits_{v=0}^{M-n-K}(-1)^v { {M-n} \choose v}  \left[\frac{{ {M-n-v}\choose K }}{{M \choose K}}\right]^{N-1},
\end{array} 
\vspace{-2px}
\end{equation}
for $n \in [\max(0,M\!-\!K(N\!-\!1)),M\!-\!K]$ and $p_n(N\!-\!1)\!=0$ otherwise.

%\newpage
Let us assume that \textit{exactly} $n \in [\max(0,M\!-\!K(N\!-\!1)),M\!-\!K]$ channels are available after $N\!-\!1$ contenders have made their decisions. Therefore, the probability of $K_f$ failures in a series of $K$ transmissions for the tagged device may be estimated as the probability to fall into already selected $M-n$ channels, while other $K-K_f$ fall into $n$ available channels. Further, $M-K_f$ replicas that have successfully avoided collisions with other contenders may be unsuccessfully received due to the error-prone channels.  Given the corruption probability of $\gamma$ per each of them, we may estimate the conditional probability that the entire message is lost as:
\vspace{-5px}
\begin{equation}
\begin{array}{c}
\!\!p_{\text{lost}|n} \!=\!\! \Pr\{\text{message lost}|n \text{ channels empty}\}\! = \!\!\\
\!\! \sum  \limits_{K_f=\max(0, K-n)}^{K}  {K\choose K_f} \gamma^{K-K_f} \left( \frac{M-n}{M} \right)^{K_f} \left( \frac{n}{M} \right)^{K-K_f} , \!\!%} ,
\end{array}\vspace{-5px}
\end{equation}
and thus the probability of successful message delivery may be approximated by:
\vspace{-5px}
\begin{equation}\vspace{-5px}
\begin{array}{c}
\!\!\!\!p_{\text{s}}(N,M,\gamma,K)\! %\sum  \limits_{K_f=0}{K}  \gamma^{K-K_f} \Pr\{K_f\text{ failures}|n \text{ channels are empty}\}\! 
=
1- \sum  \limits_{n=n_{0}}^{M-K} p_{\text{lost}|n}\cdot p_n(N-1),
%=
%1- \sum  \limits_{K_f=0}^{K}  \gamma^{K-K_f}\!\left[\frac{{M-n \choose K_f}} { {M \choose K_f}}\right] {M \choose n} \sum \limits_{v=0}^{M-n-K} (-1)^v {M-n \choose v} \left[\frac{{M-n-v \choose K}} { {M \choose K}}\right] ^{N-1}\!\!\!\!\!\!=
%\\
%\!1\!\!-\!\!\!\!\!\sum  \limits_{n=n_{0}}^{M-K} \!\!\!{{M\choose n}} \!\!\sum \limits_{v=0}^{M-n-K} \! \!\!(-1)^v {M-n \choose v}\!\! \left[\frac{{M-n-v \choose K}}{{M \choose K}}\right] ^{N\!-\!1}\!\!\! 
%\cdot \\   
%\sum  \limits_{K_f=\max(0, K-n)}^{K}  {K\choose K_f} \gamma^{K-K_f} \left( \frac{M-n}{M} \right)^{K_f} \left( \frac{n}{M} \right)^{K-K_f}  
\end{array}
\vspace{-5px}
\end{equation}
where $n_0\! =\! \max(0,M\!-\!K(N\!-\!1))$ and $p_{\text{s}}(1, M,\gamma,K) =1\!-\!\gamma^K$.

%\vspace{-5px}
%\section*{Acknowledgment}
%%
%\vspace{-5px}
%This work is supported by the project TAKE-5: The 5th Evolution Take of Wireless Communication Networks, funded by Tekes. 
%The work of the first author is supported with a personal research grant by the Finnish Cultural Foundation.

\vspace{-5px}
\bibliographystyle{ieeetr}

\bibliography{refs2}

\begin{thebibliography}{1}

\bibitem{choi2006multichannel}
Y.-J. Choi, S.~Park, and S.~Bahk, ``{Multichannel random access in OFDMA
  wireless networks},'' {\em IEEE Journal on Selected Areas in Communications},
  vol.~24, no.~3, pp.~603--613, 2006.

\bibitem{paolini2015coded}
E.~Paolini, C.~Stefanovic, G.~Liva, and P.~Popovski, ``{Coded random access:
  applying codes on graphs to design random access protocols},'' {\em {IEEE
  Communications Magazine}}, vol.~53, no.~6, pp.~144--150, 2015.

\bibitem{hajek1994delay}
B.~Hajek, N.~B. Likhanov, and B.~S. Tsybakov, ``On the delay in a
  multiple-access system with large propagation delay,'' {\em IEEE Transactions
  on Information Theory}, vol.~40, no.~4, pp.~1158--1166, 1994.

\bibitem{choudhury1983diversity}
G.~Choudhury and S.~Rappaport, ``Diversity {ALOHA--A} random access scheme for
  satellite communications,'' {\em IEEE Transactions on Communications},
  vol.~31, no.~3, pp.~450--457, 1983.

\bibitem{taghavi2016design}
A.~Taghavi, A.~Vem, J.-F. Chamberland, and K.~R. Narayanan, ``On the design of
  universal schemes for massive uncoordinated multiple access,'' in {\em IEEE
  International Symposium on Information Theory (ISIT)}, pp.~345--349, IEEE,
  2016.

\bibitem{galinina2013stabilizing}
O.~Galinina, A.~Turlikov, S.~Andreev, and Y.~Koucheryavy, ``Stabilizing
  multi-channel slotted {ALOHA} for machine-type communications,'' in {\em IEEE
  International Symposium on Information Theory (ISIT)}, pp.~2119--2123, IEEE,
  2013.

\bibitem{zanella2012estimating}
A.~Zanella, ``Estimating collision set size in framed slotted {ALOHA} wireless
  networks and {RFID} systems,'' {\em IEEE Communications Letters}, vol.~16,
  no.~3, pp.~300--303, 2012.

\bibitem{feller1968introduction}
W.~Feller, {\em {An introduction to probability theory and its applications:
  volume I}}, vol.~3.
\newblock John Wiley \& Sons, 1968.

\end{thebibliography}
\end{document}